\RequirePackage{fix-cm}

\documentclass[twocolumn]{svjour3}

\usepackage{subfig,multirow,hyperref,graphicx,amsmath}
\usepackage[authoryear]{natbib}
\smartqed 
\begin{document}

\title{Sequential Change Detection in the Presence of Unknown Parameters}
 
 \author{Gordon J. Ross}

\institute{Gordon J. Ross \at
              Heilbronn Institute for Mathematical Research\\
              University of Bristol, Bristol, United Kingdom\\
              Tel.: +123-45-678910, Fax: +123-45-678910\\
              \email{gordon.ross@bristol.ac.uk}          
    }

\date{Received: date / Accepted: date}
\maketitle

\begin{abstract}It is commonly required to detect change points in sequences of random variables. In the most difficult setting of this problem, change detection must be performed sequentially with new observations being constantly received over time. Further, the parameters of both the pre- and post- change distributions may be unknown. In \cite{Hawkins2005a}, the sequential generalised likelihood ratio test was introduced for detecting changes in this context, under the assumption that the observations follow a Gaussian distribution. However, we show that the asymptotic approximation used in their test statistic leads to it being conservative even when a large numbers of observations is available. We propose  an improved procedure which is more efficient, in the sense of detecting changes faster, in all situations. We also show that similar issues arise in other parametric change detection contexts, which we illustrate by introducing a novel monitoring procedure for sequences of Exponentially distributed random variable, which is an important topic in time-to-failure modelling.


\end{abstract}

\section{Introduction}

Change detection problems, where the goal is to monitor for distributional shifts in a sequence of time-ordered observations, arise in many diverse areas such as the segmentation of speech signals in audio processing \citep{Andre-Obrecht1988}, RNA transcription analysis in biology \citep{Caron2012}, and intrusion detection in computer networks. \citep{Tartakovsky2005,Levy-Leduc2009}. They have been especially studied within the field of statistical process control (SPC) where the goal is  to monitor the quality characteristics of an industrial process in order to detect and diagnose faults \citep{Lai1995,Hawkins1998}. 

In a typical setting, a sequence of observations $x_1$, $x_2$, $\ldots$ are received from the random variables $X_1, X_2,\ldots$. A number of abrupt change points  $\tau_1, \tau_2, \ldots$ divide the sequence into segments, where the observations within each segment are independent and identically distributed. The sequence is hence distributed as:
\begin{equation}
X_i \sim  \left\{ \begin{array}{rl}
 F_0 &\mbox{ if $i \leq \tau_1$} \\
 F_1 &\mbox{ if $\tau_1 < i \leq \tau_2$} \\
 F_2 &\mbox{ if $\tau_2 < i \leq \tau_3$}, \\
\ldots
       \end{array} \right.
\label{eqn:framework}
\end{equation}
for some set of distributions $\{F_0, F_1, \ldots\}$. The goal is  to estimate the location of the change points. Although the assumption of independent observations between change points may seem restrictive, this is not the case since a statistical model can usually be fitted to the observations to model any dependence, with change detection then being performed on the independent residuals. An extensive discussion of this topic can be found in \cite{Gustafsson2000}.

There are two  different versions of the change detection problem. In the \textbf{batch} version, the sequence has a fixed length consisting of $n$ observations. Change detection is performed retrospectively using the whole sequence at once (e.g. \cite{Hinkley1970a}). In the \textbf{sequential} version, the sequence does not necessarily have a fixed length.  Instead, observations are received and processed in order. The sequence is monitored for changes and, after each observation has been received, a decision is made about whether a change has occurred based only on the observations which have been received so far \citep{Lai1995}. If no change is flagged, then the next observation in the sequence is processed, and so on. The rate at which observations arrive imposes computational constraints on change detection algorithms, with a typical requirement being that algorithms should be at worst O(n), but preferably O(1). These two settings are known in the SPC literature as Phase I and Phase II respectively, and Phase II change detection algorithms are commonly referred to as control charts.

Our concern is the sequential Phase II setting. One advantage of this setting is that it only requires a single change point to be detected at any given time, which avoids most of the computational complexity associated with detecting multiple change points. We  therefore refer to $F_0$ and $F_1$ as the pre- and post-change distributions respectively, with the single change point being denoted as $\tau$.  

In typical SPC applications it can be assumed that the parametric forms of $F_0$ and $F_1$ are known, and the Gaussian case where $F_0 = N(\mu_0, \sigma^2_0)$ and $F_1 = N(\mu_1,\sigma^2_1)$ is of particular interest for SPC.  Many existing procedures focus only on monitoring for changes  in the mean of such a sequence, however we concern ourselves with the more general case where either the mean and variance may undergo change. Existing approaches for this problem differ in what is assumed to be known about the pre- and post- change means and variances. In the utopian case where all these parameters are known exactly, the minimax optimal sequential change detection procedure is the well-known CUSUM chart \citep{Hawkins1998}. However in most realistic settings these parameters are unknown, and \cite{Jensen2006} showed that naive attempts to estimate them can lead to poor performance. 

Until recently, the standard frequentist approach when working with unknown pre-change parameters was the self starting CUSUM chart discussed in \cite{Hawkins1998}, which adapts the CUSUM to situations where distributional parameters are unknown. However, this CUSUM chart suffers from requiring knowledge of the post-change parameters, which is generally unavailable. To alleviate this, \cite{Hawkins2005a}  (subsequently referred to as HZ) recently proposed a new approach to sequential change detection for Gaussian based on a repeated series of generalised likelihood ratio tests. Their approach was shown to perform favourably compared to the self starting CUSUM, and can thus be considered the current state-of-the-art for frequentist sequential Gaussian monitoring.

In this paper, we present a new change detection algorithm which improves on their proposal. The test statistic used in HZ relies on maximising over a collection of likelihood ratio statistics, which are each marginally assumed to be approximately  $\chi^2_2$ distributed. However this approximation only holds asymptotically. Due of the peculiarities of the sequential change detection setting, this asymptotic result is never achieved even when the number of available observations is very large. This results in their procedure being somewhat conservative, which reduces its ability to detect changes quickly. This is undesirable, since changes must be detected as fast as possible in typical SPC application. We introduce a different statistic which avoids this problem, and results in quicker detection of every type of  change. 

We also discuss how the framework introduced by HZ can be used for parametric monitoring in more general situations where the distributional form of $F_0$ and $F_1$ may be non-Gaussian. In this setting the same problem relating to the failure of asymptotic assumptions also arises, and performance can again be improved if a suitable correction is made. We illustrate this phenomena by introducing a novel statistic for detecting changes in sequences of Exponentially distributed random variables. Control charts for the Exponential distribution are of interest to SPC due to their use in monitoring the time between failures generated by high yield processes \citep{Khool2009, Liua2007, Chan2000}, and our proposal extends this work by not requiring prior knowledge of either the pre- or post- change distributional parameter.

The remainder of the paper proceeds as follows: in Section \ref{sec:detection} we summarise the sequential change detection framework of HZ. Then, in Section \ref{sec:finite} we discuss some limitations of their approach which results in slower change detections. Based on this, we formulate a new test statistic which corrects these issues. Section \ref{sec:exponential} introduces a new statistic for detecting changes in sequences of Exponential random variable. Finally Section \ref{sec:experiments} compares the performance of the test statistics with and without finite sample corrections,  across a range of change detection scenarios and also includes a comparison to recently proposed Bayesian methods for sequential change detection.

\section{Change Detection}
\label{sec:detection}

The procedure described in HZ extends the standard Phase I generalised likelihood ratio test of \cite{Hinkley1970a}  to sequential monitoring. First consider the (non sequential) Phase I setting where there is a fixed size sample of $n$ Gaussian observations $x_1,\ldots,x_n$ containing at most one change point. To test whether a change point occurs at some particular location $\tau=k$, the observations are divided into two samples $\{x_1,\ldots,x_k\}$ and $\{x_{k+1},\ldots,x_n\}$, and a likelihood ratio test is used to assess whether these samples have equal  means and variance. In this case the null hypothesis is:
\[H_0: X_i \sim N(\mu_0, \sigma_0^2) \quad \forall i.\]
while the alternative hypothesis that there is a single change point at  $k$ is:
\[H_1: X_i \sim  \left\{ \begin{array}{rl}
N(\mu_0, \sigma_0^2) &\mbox{ if $i \leq k$} \\
N(\mu_1, \sigma_1^2) &\mbox{ if $i > k$.} \\
       \end{array} \right.
\]

In both cases, the parameters $\mu_0$, $\sigma_0^2$, $\mu_1$ and $\sigma_1^2$ are unknown and must be estimated. Writing $L_0$ and $L_1$ for the respective likelihoods under the null and alternative hypothesis, and letting $D_{k,n} = -2\log(L_0/L_1)$, the standard likelihood ratio test statistic in this situation can be written as:
\begin{equation}
D_{k,n} =  k  \log \frac{S_{0,n}}{S_{0,k}} + (n-k) \log  \frac{S_{0,n}}{S_{k,n}},
\label{eqn:lrt}
\end{equation}
where:
\[S_{r,s} = \sum_{i=r+1}^{s} (x_i - \bar{x}_{r,s})^2/(s-r), \quad \bar{x}_{r,s} = \sum_{i=r+1}^{s} x_i/(s-r),\]
and we note that $S_{r,s}$ is the (biased) maximum likelihood estimate of the variance. Under the null hypothesis, observations from both samples are identically distributed and $D_{k,n}$ has an asymptotic chi-square distributed with two degrees of freedom. This allows critical values to be computed, with the null hypothesis being rejected if $D_{k,n}$ exceeds a given value. 

Of course, in practice it will not be known which value of $k$ to use as the change point location in the above test. Therefore, it is usual to treat $k$ as a nuisance parameter and estimate it via maximum likelihood. This leads to the following generalised likelihood ratio test statistic, which tests whether a change occurs at any point in the sequence:
\begin{equation}
D_n = \max_k D_{k,n}, \quad 2 \leq k \leq t-2.
\label{eqn:maximisation}
\end{equation}

It is concluded that the sequence contains a change point if $D_n > h_n$ for some appropriately chosen threshold $h_t$. The estimate of the change point is then the value of $k$ for which $D_{k,n}$ is maximal.

This procedure assumes that the sample $x_1,\ldots,x_n$ has a fixed length. However in many applications, such as those commonly encountered in Phase II SPC,  this is not the case and  new observations are  received over time. In this case monitoring for a change must be performed sequentially,]with a decision about whether a change has occurred being taken after every observation. The above framework can be extended to this situation by processing the observations sequentially, starting with the first. For each observation $x_t$, the statistic $D_t$ is computed using only this observations and the previous ones, i.e. $x_1,\ldots,x_t$. If $D_t > h_t$ then a change is flagged, otherwise the next observation $x_{t+1}$ is processed, and $D_{t+1}$ is computed, and so on. This procedure is hence a sequence of generalised likelihood ratio tests, and allows  sequences containing multiple change points to be processed without a high computational burden, assuming that previous observations are discarded and the procedure is restarted whenever a change is detected.

In order for this procedure to be feasible, it must be possible to compute the $D_t$ statistics without incurring too great a computational cost. As discussed in \cite{Hawkins2005a}, the likelihood ratio statistics $D_{k,t}$ can be written in terms of the sufficient statistics for estimating the mean and variance of a Gaussian distribution, and these admit a simple recursively updatable form. Therefore, computing $D_{k+1,t}$ given $D_{k,t}$ can be performed very fast. One problem which can arise is that the number of likelihood ratio tests performed when each observation is processed grows linearly over time, since the calculation of $D_t$ requires the computation of $D_{2,t}$, $D_{3,t}$,$\ldots$, $D_{t-2,t}$. This means that the algorithm will have a quadratic $O(n^2)$ time complexity. However, a windowing procedure can be used to drastically reduce the number of statistics which must be computed, giving an $O(n)$ algorithm. Crucially, this windowing procedure need not lead to any drop in performance, since older observations can be incorporated into fixed size summery statistics rather than being discarded. This is described in more detail in \cite{Hawkins2005a}.

The other key issue is determining the sequence of thresholds $\{h_t\}$ in a way which takes into account that multiple highly correlated tests are being performed. The procedure used by HZ is to choose this sequence so that the probability of incurring a false positive is constant over time, i.e. assuming that no change has occurred, choose the thresholds so that:

\begin{equation}
\begin{split}
P(D_t  > h_t | D_{t-1} \leq h_{t-1}, & D_{t-2}  \leq h_{t-2}, \\& \ldots, D_1 \leq h_1) =  \gamma, \quad \forall t.
\end{split}
\label{eqn:conditional}
\end{equation}
In SPC, it is common to design change detection algorithms such that, assuming there has been no change, there is a hard bound on the expected number of observations until a false positive is signaled. The expected number of observations before a signal is known as the Average Run Length ($ARL_0$) and in this case it is clear that the $ARL_0$ is equal to $1/\gamma$. However, the analytic form of the marginal distribution of $D_t$ is only known asymptotically, and the conditional distribution in Equation \ref{eqn:conditional} is more complex and does not have a known form, even asymptotically. Therefore, a Monte Carlo procedure can instead be used to generate the thresholds. By simulating several million sequences of independent $N(0,1)$ random variables, the thresholds corresponding to various choices of $\gamma$ can be determined empirically. Although this procedure is computationally expensive, it need only be carried out a single time. As the null distribution of $D_{k,t}$ is independent of the unknown mean and variance of the observations, the computed thresholds will give the required $ARL_0$ for any Gaussian sequence, and can hence be stored ahead of time in a lookup table, so that no extra computational cost is added when processing the sequence. We discuss this procedure further in the following section, after first describing our new test statistic.

\section{Finite Sample Correction}
\label{sec:finite}

\begin{figure*}[t]
  \centering
  \subfloat[Values of $H_{k,50}$ (dotted line) and $D^c_{k,50}$ (solid line)]{\label{fig:finite1}\includegraphics[width=0.45\textwidth]{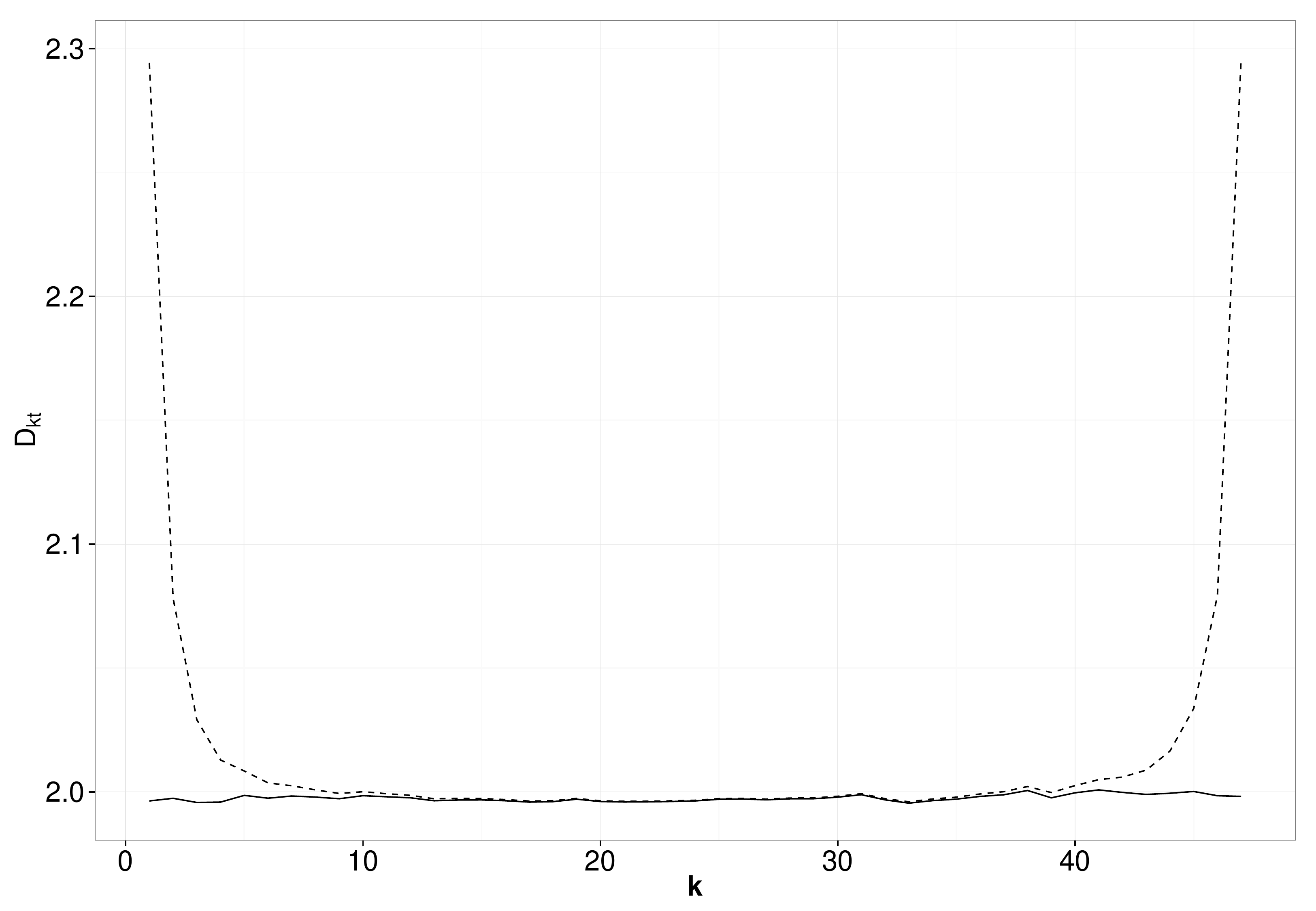}}
  \subfloat[Values of thresholds $h_t$ for HK chart (dotted line) and our proposal (solid line)]{\label{fig:finite2}\includegraphics[width=0.45\textwidth]{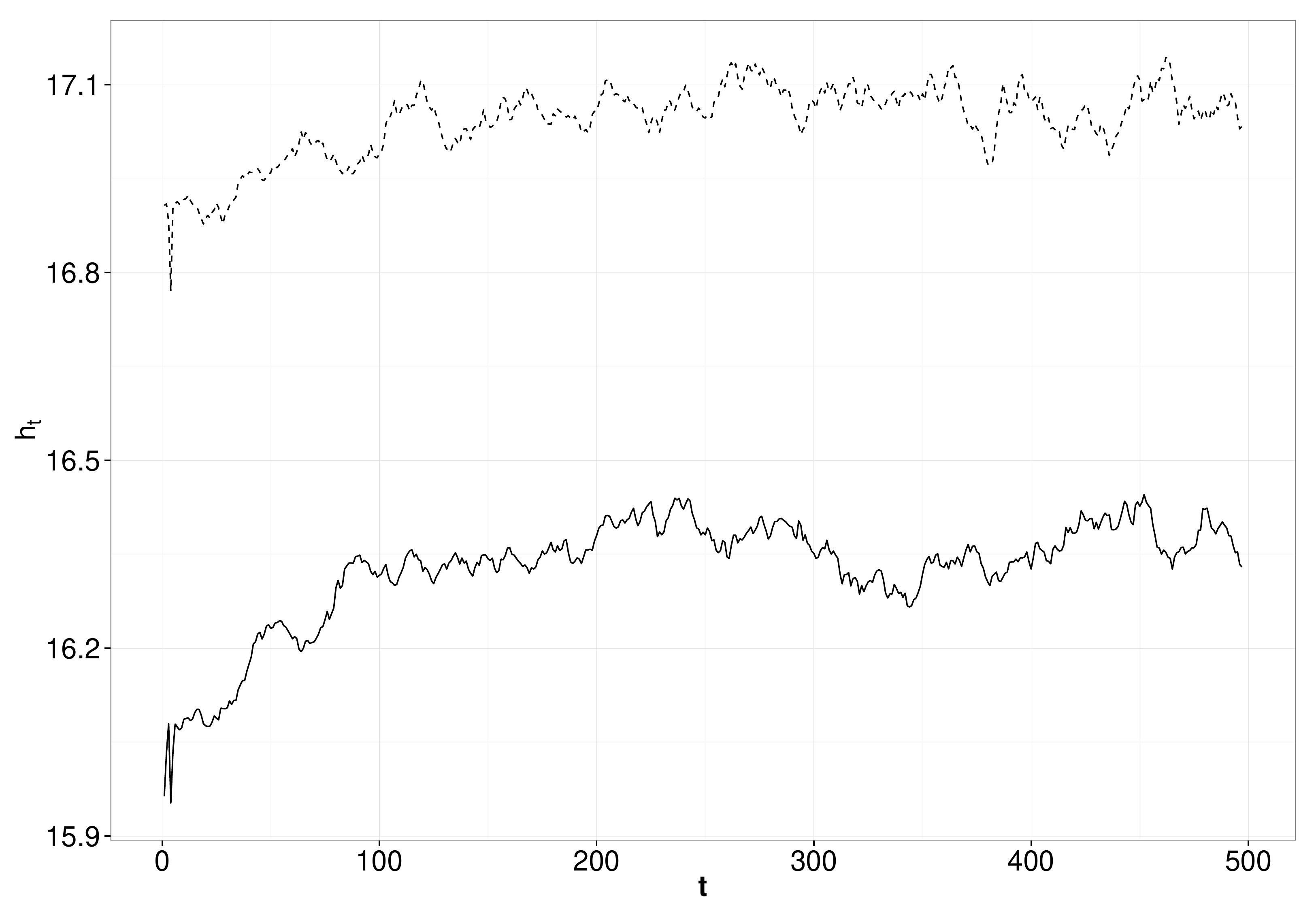}}
  \label{fig:mds}
\end{figure*}

A limitation of the  above HZ procedure is that, while the likelihood ratio test statistics $D_{k,t}$ in Equation \ref{eqn:lrt} are each asymptotically $\chi_2^2$ distributed under the null hypothesis as the size of both samples grows, their distribution will differ from this in a finite sample setting. This is problematic since $D_t$ is defined by maximising over $D_{k,t}$, which means that if some values of $D_{k,t}$ have a higher mean and/or variance than others, they will tend to dominate the maximisation. This can lead to the thresholds $h_t$ being artificially high, which reduces the power of the test and, in a sequential context, increases the length of time taken to detect change points. To reduce the impact of this, HZ make use of a Bartlett correction  in order to increase the rate at which the $D_{k,t}$ statistics converge to $\chi_2^2$. Bartlett corrections are motivated by the well known result that dividing a likelihood ratio test statistic by its expected value under the null hypothesis gives a transformed statistic which converges to $\chi_2^2$ at a faster rate. In HZ, an approximate correction is used where the test statistic is divided by a constant $C_{k,t}$, leading to a new test statistic $H_{t}$ where
\[H_t =  \max_k H_{k,t}, \quad H_{k,t} = D_{k,t} /C_{k,t},\]
and $C_{k,t}$ is the Bartlett correction factor:
\[C_{k,t} = 1 + \frac{11}{12}\left(\frac{1}{k} + \frac{1}{t-k} - \frac{1}{t}\right) + \left(\frac{1}{k^2} + \frac{1}{(t-k)^2} - \frac{1}{t^2}\right).\]

However, this does not fully resolve the issue. Ordinarily, such an approximate correction would result in a test statistic which is approximately $\chi_2^2$ distributed for moderate sized samples, but in the change point setting, this is not the case. The problem is that the maximization over $k$ in Equation \ref{eqn:maximisation} leads to $D_{k,t}$ being computed for small values of $k$, such as $k < 5$ (or symmetrically when $k > t-4$), even when the number of observations $t$ is large. Therefore, one sample will contain a very small number of observations regardless of how large $t$ is, and the test statistic may hence differ substantially from its asymptotic distribution. To quantify this, consider the expected value of $H_{k,t}$ under the null hypothesis when a large number of observations have been received and $t \rightarrow \infty$ but $k$ remains small. It can be shown (see Appendix) that:

\begin{align*}
E[D_{k,t}] =  t& (\log(2/t) + \psi( (t-1)/2)) \\& - k(\log(2/k) + \psi((k-1)/2)) \\& - (t-k)(\log(2/(t-k)) + \psi((t-k-1)/2)),
\end{align*}
where $\psi(z) = \Gamma ' (z) / \Gamma(z)$ is the digamma function. For large values of $t$, $E[D_{k,t}]  \rightarrow_t \log(2/k) + \psi( (k-1)/2)$. Similarly, the value of the Bartlett correction factor $C_{k,t}$ as $t \rightarrow \infty$ is $1 + 11/12 k^{-1} + k^{-2}$ hence (by Slutsky's theorem):
\[E[H_{k,t}] \rightarrow_t  \frac{\log(2/k) + \psi( (k-1)/2)}{1 + 11/12 k^{-1} + k^{-2}}.\]

A $\chi_2^2$ random variable has an expected value of $2$, and for $k \in \{2,3,4,5,6\}$ the corresponding asymptotic values of $E[H_{k,t}]$ are $\{2.30, 2.08,  2.03, 2.02, 2.01\}$ which shows that the approximation fails for small values of $k$. This means that the $H_{t}$ statistic used in HZ has a substantially higher expected value under the null hypothesis than it should otherwise have, due to these large expected values at the boundary values of $k$. This leads to their threshold sequence $h_t$ being inflated, which can cause changes to be detected slower, since the higher $h_t$ thresholds take longer to be breached after a change occurs.

Figure \ref{fig:finite1}  illustrates this effect by plotting the expected values of $H_{k,t}$ when $t=50$. The spike at the end when $k<5$ (and symmetrically when $k > 45$) is clearly visible. One possible way to solve this problem would be to only carry out the maximization of $H_{k,t}$ over a smaller range of values, such as $5 \leq k \leq t-4$. However this is not practical in the sequential setting. When a change occurs in the sequence, it must generally be detected as quickly as possible. This means that ideally the number of post-change observations that need to be processed before the change is detected should be small, which implies that the largest $H_{k,t}$ values should occur when only a small number of post-change observations are split off.  By not including the $k > t-4$ terms in the maximization, the ability to detect changes quickly is hence reduced, and these slower detections can be a serious problem in practice.

We therefore instead propose using a new set of statistics $D^c_{k,t}$ for change detection using a better correction to the likelihood ratio test. It is well known (for example \cite{Jensen1993}) that if $\Lambda$ denotes a log likelihood ratio test with an asymptotic $\chi_q^2$ distribution, then convergence can be improved by instead working with $(q\Lambda)/E[\Lambda]$, which converges to $\chi_q^2$ at a faster rate. This motivates using the following finite-sample corrected test statistics:

$$D^c_{k,t} = \frac{2D_{k,t}}{E[D_{k,t}]}, \quad D^c_t = \max_k D^c_{k,t}$$

where $E[D_{k,t}]$ is defined as above. Figure \ref{fig:finite1} shows the expected values of these $D^c_{k,t}$ statistics when $t=50$, and it can be seen that the small sample spike no longer exists at the boundaries. Using this new $D^c_{k,t}$ statistic, we computed threshold sequences $h_t$ corresponding to various values of the $ARL_0$ by simulation, using the Monte Carlo approach described in the previous section. For several different choices of the $ARL_0 = 1/\gamma$, we generated $2$ million random sequences of Gaussian variables and computed empirically the values of $h_t$ which would give such an $ARL_0$. In order to reduce the sampling variation that occurs when simulating such threshold sequences, the sequences are then exponentially smoothed using the formula $\tilde{h}_t = 0.7\tilde{h}_{t-1} + 0.3h_t$. The Appendix contains examples of such sequences for some of the more commonly used $ARL_0$ values; note that we have restricted the monitoring procedure to begin after the $20^{th}$ observation, on the grounds that when fewer observations are available it becomes increasingly difficult to detect changes. 

However, working with such raw sequences is cumbersome, so in order to allow practitioners to more easily use our algorithm we provide the following approximate equation relating $h_t$ to various values of $\gamma$, which was found by fitting a non-linear regression model to the generated sequences:
\begin{equation}
h_t =    1.51 - 2.39 \log(\gamma) + \frac{3.65 + 0.76\log(\gamma)}{\sqrt{t-7}}. 
\end{equation}
It is interesting to compare the value of these thresholds to the  thresholds when using the HZ procedure.  In Figure \ref{fig:finite2}, the dotted line shows the values of the $h_t$ statistic corresponding to an $ARL_0$ of $500$ for HZ, while the $h_t$ values for the same $ARL_0$  when using our statistic are plotted as a solid line. It can be seen that the thresholds when using our statistic are lower than when using that of HZ,  since they are not being distorted by the extreme values associated with the $k < 5$ cases. We will show in Section \ref{sec:experiments}  that this leads to the faster detection of all types of change.

\section{Change Points in Exponentially Distributed Sequences}
\label{sec:exponential}


Although the above discussion has focused on detecting changes in Gaussian sequences, the general change point model framework can be used for change detection in other parametric contexts where unknown parameters are present. In most cases the approach will be identical to the above, with a test statistic derived from the two-sample likelihood ratio test being maximized over every possible split point in the sequence, As in the above Gaussian discussion, there may also be a need to make a finite sample corrections to the test statistics in order to prevent large values in the small segments distorting the maximization. To illustrate this, we now develop a test statistic to monitor for changes in a sequence of Exponential random variables when both the pre- and post-change parameters are unknown. As discussed in the introduction, control charts for the Exponential distribution are widely used to monitor for changes in the expected time between failures in SPC situations.

As in the previous section,  we begin with a Phase I setting with $n$ Exponentially distributed observations $x_1,\ldots,x_n$ containing at most one change point. To test whether a change point occurs at location $\tau = k$, the sequence is split into the two samples $\{x_1,\ldots,x_k\}$ and $\{x_{k+1},\ldots,x_n\}$. The null hypothesis of identical distribution is then:

\[H_0: X_i \sim \mathrm{Exp} (\lambda_0) \quad \forall i.\]

The alternative hypothesis that there is a single change point at location $k$ is:
\[H_1: X_i \sim  \left\{ \begin{array}{rl}
\mathrm{Exp}(\lambda_0) &\mbox{ if $i \leq k$} \\
\mathrm{Exp}(\lambda_1) &\mbox{ if $i > k$.} \\
       \end{array} \right.
\]

where $\lambda_0$ and $\lambda_1$ are unknown. Letting $L_0$ and $L_1$ denote the likelihoods under the null and alternative hypothesis respectively, and writing $M_{k,n} = -2 \log(L_0/L_1)$, the generalized likelihood ratio test statistic is then:

\begin{align*}
M_{k,n} =  -2\left(n \log \frac{n}{S_{0,n}} - k \log \frac{k}{S_{0,k}} - (n-k) \log \frac{n-k}{S_{k.n}}\right)
\end{align*}

where $S_{i,j}$ is defined as above. The test for a change point could then be based on $M_{n} = \max_k M_{n,k}$. However doing this naively leads to the same problem as in the Gaussian case, namely that even though $M_{k,n}$ has an asymptotic $\chi^2_1$ distribution, this will not be achieved when $k$ is close to $0$ or $n$ . Specifically, it can be shown (see Appendix) that:

\begin{figure}
  \centering
  \includegraphics[width=0.4\textwidth]{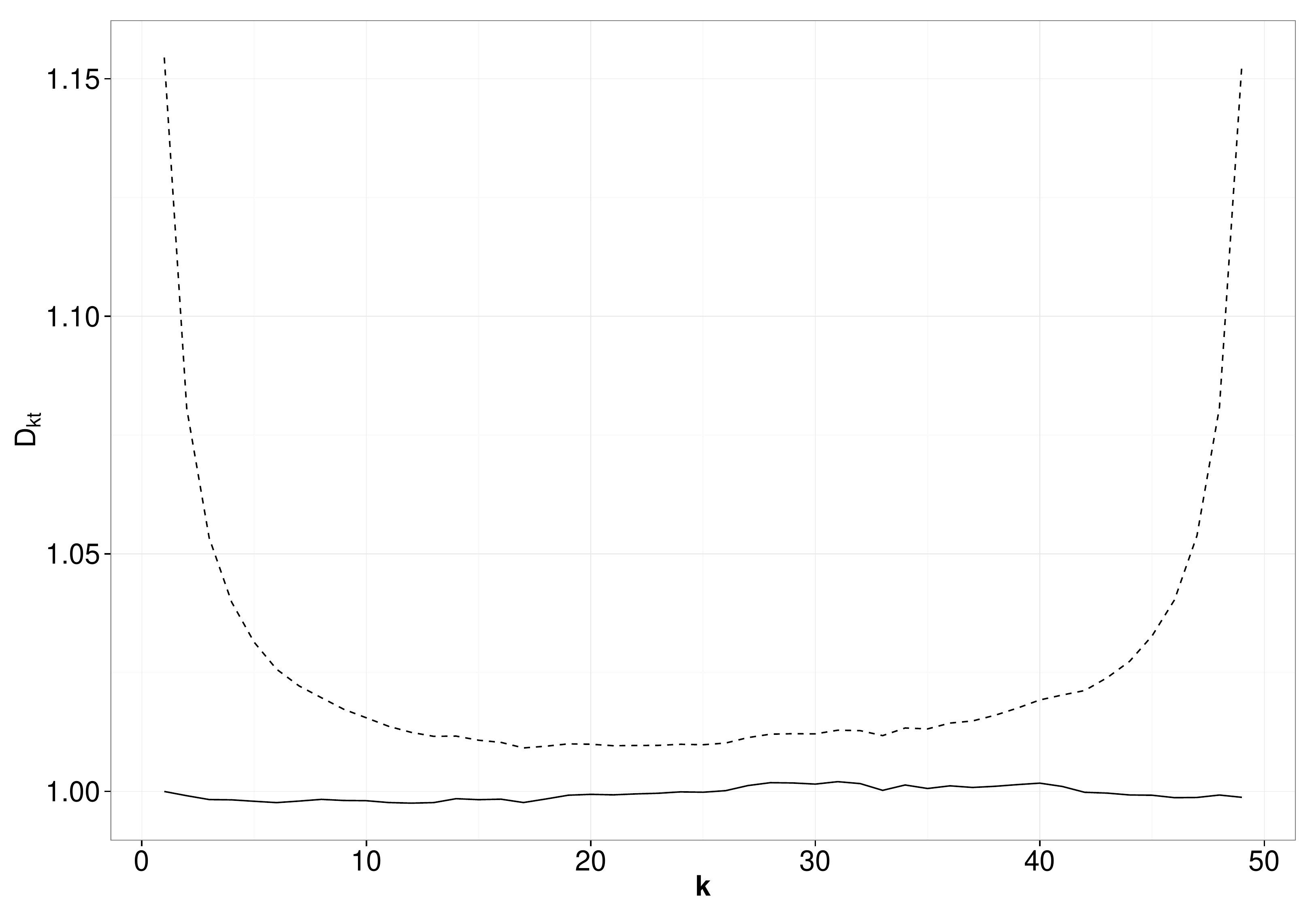}
 \caption{Values of $M_{k,50}$ (dotted line) and $M^c_{k,50}$ (solid line)}
\label{fig:expfinite1}
\end{figure}

\begin{align*}
E[M_{k,t}] =   -2[&k\psi(k) + (n-k)\psi(n-k) -n\psi(n) + \\& n \log (n) - k \log k - (n-k) \log (n-k)].
\end{align*}

Hence as $n \rightarrow \infty$, $E[M_{n,k}] \rightarrow -2k[\psi(k) - \log(k)]$ for any fixed $k$, while the expected value of a  $\chi^2_1$ random variable is $1$ .  Figure \ref{fig:expfinite1} illustrates this by plotting the expected values of $M_{k,50}$ for each value of $k \in \{1,2,\ldots,49\}$, which shows the same pattern as before, with a large spike as both boundaries are approached. 

 In order to correct this, we again define a corrected version of the test statistic by again dividing by its  finite sample expectation:

$$M^c_{k,t} = \frac{M_{k,t}}{E[M_{k,t}]}, \quad M^c_t = \max_k M^c_{k,t}.$$

The values of the $M^c_{k,50}$ statistics are also plotted on Figure  \ref{fig:expfinite1} and it can be seen that the spike at the boundary no longer exists.

Sequential change detection can then be carried out in the same manner to the Gaussian case, with $M_t^c$ being recomputed after each observation, and the $h_t$ sequence being chosen to bound the $ARL_0$. Table \ref{tab:thresholds} in the Appendix gives the values of $h_t$ which correspond to various values of the $ARL_0$ for both $M_t^c$ and $M_t^c$. In the next section we will show how using the finite-sample correction again leads to superior change detection performance.




\section{Performance Analysis}
\label{sec:experiments}

We  now investigate the performance of the proposed change detection statistics. First, in Section \ref{sec:experiments1} we compare the performance of the statistics which have had their finite sample moments corrected, to the uncorrected versions which rely on asymptotic distributions. By the arguments in the previous section, it should be expected that the corrected versions outperform the uncorrected versions in all reasonable situations. Next, in Section \ref{sec:experiments2} we discuss how the frequentist paradigm used in this paper compares to recent Bayesian approaches for sequential change detection, such as that of \cite{Fearnhead2007}. Finally in Section \ref{sec:experiments3} we explore how the statistics perform when applied to several real data sets.

\subsection{The Effect of the Finite Sample Correction}

\label{sec:experiments1}

We begin by comparing the finite-sample corrected Gaussian and Exponential change detection statistics to both the Gaussian statistic from \cite{Hawkins2005a} which uses an asymptotic correction, and to the uncorrected Exponential statistic from Section \ref{sec:exponential}. As discussed in Section \ref{sec:detection}, the standard approach for comparing the performance of frequentist change detection methods is to ensure that each method generates false positives at the same rate (denoted by the $ARL_0$)  under the assumption that there has been no change points,  and to then compare the average number of post-change observations required before changes of various magnitudes are detected \citep{Basseville1993}. This is analogous to comparing classical hypothesis tests where the power of each test is investigated subject to a bound on the Type I error probability. In the below experiments we have chosen a value of $ARL_0 = 500$ for each chart, although the patterns we observe are the same for all values.

Since we are concerned with cases where the pre-change parameters are unknown, the number of observations which are available from the pre-change distribution will affect the delay until the change is detected, as a larger number of observations means that the unknown parameters will be more accurately estimated, resulting in quicker detection. We therefore investigate changes which occur at locations $\tau = 25$ and $\tau = 100$ which correspond to an early and a late change respectively.

\subsubsection{Gaussian Sequences}

For each change point location $\tau$, the pre-change distribution is set to $X_i \sim N(0,1)$ when $i \leq \tau$. We then investigate both mean changes where the post change distribution shifts to $N(\mu_1,1)$ for $\mu_1 \in [0,2]$, and variance changes where the post change distribution shifts to $N(0,\sigma_1^2)$ for $\sigma_1 \in [0.33, 3]$. For each change location and change magnitude, $100000$ sequences were generated according to these distributions. For each sequence, the observations were processed sequentially until a change was detected. This allows the average detection delay $E[T - \tau | T > \tau]$ to be estimated, where T denotes the observation after which a change is first signalled. Ideally this delay should be as low as possible; i.e changes should be detected as soon after they occur as possible.
\begin{table}[t]
\caption{Average number of observations before a change from $N(0,1)$ to $N(\mu_1,1)$ occurring at time $\tau$ is detected. $H_t$ represents the chart from \cite{Hawkins2005a}, while $D^c_t$ denotes our the finite-sample corrected statistic.}
\begin{center}
\subfloat[$\tau=25$]{
\begin{tabular}{rrr}
  \hline
$\mu_1$ & $H_t$ & $D^c_t$ \\ 
  \hline
  0.00 & 502.4 & 497.7 \\
 0.25 &473.9 &436.0\\
        0.50& 388.0 &335.4\\
     0.75 & 211.6 & 182.4\\
        1.00  &75.7  &63.8\\
1.25  &26.0 & 23.3\\
1.50 & 13.8 &  12.8\\
    1.75   &9.7   &9.1\\
 2.00 &  7.6  &7.2\\
   \hline
\label{tab:mean50}
\end{tabular}
}
\subfloat[$\tau=100$]{
\begin{tabular}{rrr}
  \hline
$\delta$ & $H_t$ & $D^c_t$ \\ 
  \hline
  0.00 & 499.1& 504.2 \\
 0.25 &388.0 &341.9\\ 
     0.50 &118.8 &106.6\\ 
           0.75  &34.8  &32.5\\ 
   1.00  &18.5  &17.5\\ 
           1.25  &12.1  &11.6\\
            1.50  & 8.8   &8.6\\ 
             1.75   &6.9  & 6.7\\ 
   2.00   &5.6&   5.5\\ 
   \hline
\label{tab:mean300}
\end{tabular}
}
\end{center}
\end{table}

Tables \ref{tab:mean50} and \ref{tab:mean300} show the average detection delay for mean changes which occur after $25$ and $100$ observations respectively, while Tables \ref{tab:var50} and \ref{tab:var300} give the same information for variance changes. Unsurprisingly, these results show that changes which occur after $100$ observations are detected faster, with larger changes being easier to detect than smaller ones. Similar to the findings which have been reported by others who have studied changes in variance \citep{Hawkins2005a,Rosstechnometrics}, decreases in the variance take longer to detect than increases. 

Looking at the comparative performance of the finite-sample corrected $D^c_t$ statistic compared to the HZ approach, it can be seen that the former detects changes faster across every combination of change magnitude and change time. The size of the performance gain depends on the magnitude of the change;  when $\tau=25$, using the $D^c_t$ statistic results in changes being detected around $10\%$ faster, with greater improvements for smaller change magnitudes. For example, when the mean increases by $0.25$, our chart on average detects the change roughly $40$ observations faster than the HZ proposal. Such performance improvements may be substantial in typical situations where the change must be detected as fast as possible. 

Finally, we note that using the $D^c_t$  statistic still results in faster change detection when  $\tau=100$. In this case the performance improvement is slightly less substantial, particularly when the change magnitude is large. However smaller changes are still detected more than $10$ observations faster when using  $D^c_t$. This again may be a substantial improvement in a process control setting where quick detection is paramount. This highlights the previous point that the use of a larger sample does not fix the issues which negatively affect the performance of the HZ chart, since it is still constrained by having the maximization occur over split points which produce small samples. As the extra computation required to compute the $D^c_t$ statistic is minimal, we would hence recommend using $D^c_t$ in any practical situation.

\begin{table}[t]
\caption{Average number of observations before a change from $N(0,1)$ to $N(0,\sigma_1^2)$ occurring at time $\tau$ is detected. increases in variance are given first, followed by decreases.}
\begin{center}
\subfloat[$\tau=25$]{
\begin{tabular}{rrr}
  \hline
$\sigma_1$ & $H_t$ & $H^c_t$ \\ 
  \hline
  0.00 & 495.2 & 497.4 \\
1.50& 414.0 &366.4\\ 
2.00 &153.3  &124.5\\ 
 2.50  &29.0  &24.0\\ 
      3.00  &10.8   &9.9\\ 
  \hline
   0.67&294.1 &256.2\\ 
   0.50 &66.0  &57.2\\ 
     0.40 &22.5 & 20.6\\ 
       0.33 &14.5&  13.6\\ 
   \hline
\label{tab:var50}
\end{tabular}
}
\subfloat[$\tau=100$]{
\begin{tabular}{rrr}
  \hline
$\sigma$ & $H_t$ & $H^c_t$ \\ 
  \hline
  0.00 & 501.2 & 497.7 \\
       1.50 &85.3 &76.3\\ 
       2.00 &15.7 &15.0\\ 
    2.50  &8.4  &8.2\\ 
       3.00  &5.8  &5.7\\ 
  \hline
   0.67 &78.5 &72.0\\ 
        0.50 &23.8 &22.5\\ 
   0.40 &15.4 &14.7\\ 
     0.33 &12.0& 11.5\\ 
   \hline
\label{tab:var300}
\end{tabular}
}
\end{center}
\end{table}

\subsubsection{Exponential Sequences}

We perform a similar set of experiments for sequences which have an Exponential distribution where $X_i \sim \mathrm{Exp}(1)$ if $i \leq \tau$ and $X_i \sim \mathrm{Exp}(\delta)$ if $i > \tau$, where $\delta \in [0,3]$. As mentioned previously, such change detection tasks may arise in failure-time monitoring problems, or when testing for shifts in the rate of a Poisson process. 

Tables \ref{tab:exp50} and \ref{tab:exp100} show the average number of observations before a change is detected using both the finite-sample corrected and uncorrected statistics, for $\tau = 25$ and $\tau=100$. Similar to the Gaussian case, the finite sample corrected statistic has superior change detection performance across all values of $\delta$ and $\tau$, illustrating the importance of using a finite-sample correction. Again, this may prove to be very important in situations where fast change detection is critical.
\begin{table}[t]
\caption{Average number of observations before a change from $\mathrm{Exp}(1)$ to $\mathrm{Exp}(\delta)$ occurring at time $\tau$ is detected. $M_t$ represents the uncorrected test statistic, while $M^c_t$ denotes the finite-sample corrected statistic}

\begin{center}
\subfloat[$\tau=25$]{
\begin{tabular}{rrr}
  \hline
$\delta$ & $M_t$ & $M^c_t$\\ 
  \hline

            1.50 &332.1 &330.3\\
              2.00  &134.5  &127.2\\
          2.50  &47.4  &44.3\\
              3.00  &22.4  &21.2\\
              \hline
            0.67 &428.9 &418.4\\
             0.50  &224.5  &208.8\\
        0.40  &77.4 & 70.2\\
         0.33  &26.1&  24.4\\
   \hline
\label{tab:exp50}
\end{tabular}
}
\subfloat[$\tau=100$]{
\begin{tabular}{rrr}
  \hline
$\delta$ & $M_t$ & $M^c_t$ \\ 
  \hline

          1.50 &130.0 &125.5\\
          2.00  &30.1  &29.5\\
         2.50  &17.6  &17.0\\
            3.00  &12.9  &12.6\\
            \hline
              0.67 &167.3 &160.1\\
           0.50 & 27.4  &26.3\\
             0.40 & 13.4 & 13.2\\
              0.33  & 9.1&   8.9\\
   \hline
\label{tab:exp100}
\end{tabular}
}
\end{center}
\end{table}

\subsection{Comparison to Bayesian Change Detection}
\label{sec:experiments2}

This paper has focused on the frequentist paradigm, where change detection is carried out subject to a hard bound on the rate at which false positive detections occur, represented by the $ARL_0$. Due to the  structure of the likelihood ratio test statistics we considered for the Gaussian and Exponential distribution, such a bound can be guaranteed even when the pre- and post-change parameter values are unknown.

\begin{figure*}
  \centering     
  \subfloat{\label{fig:priors}\includegraphics[width=0.35\textwidth]{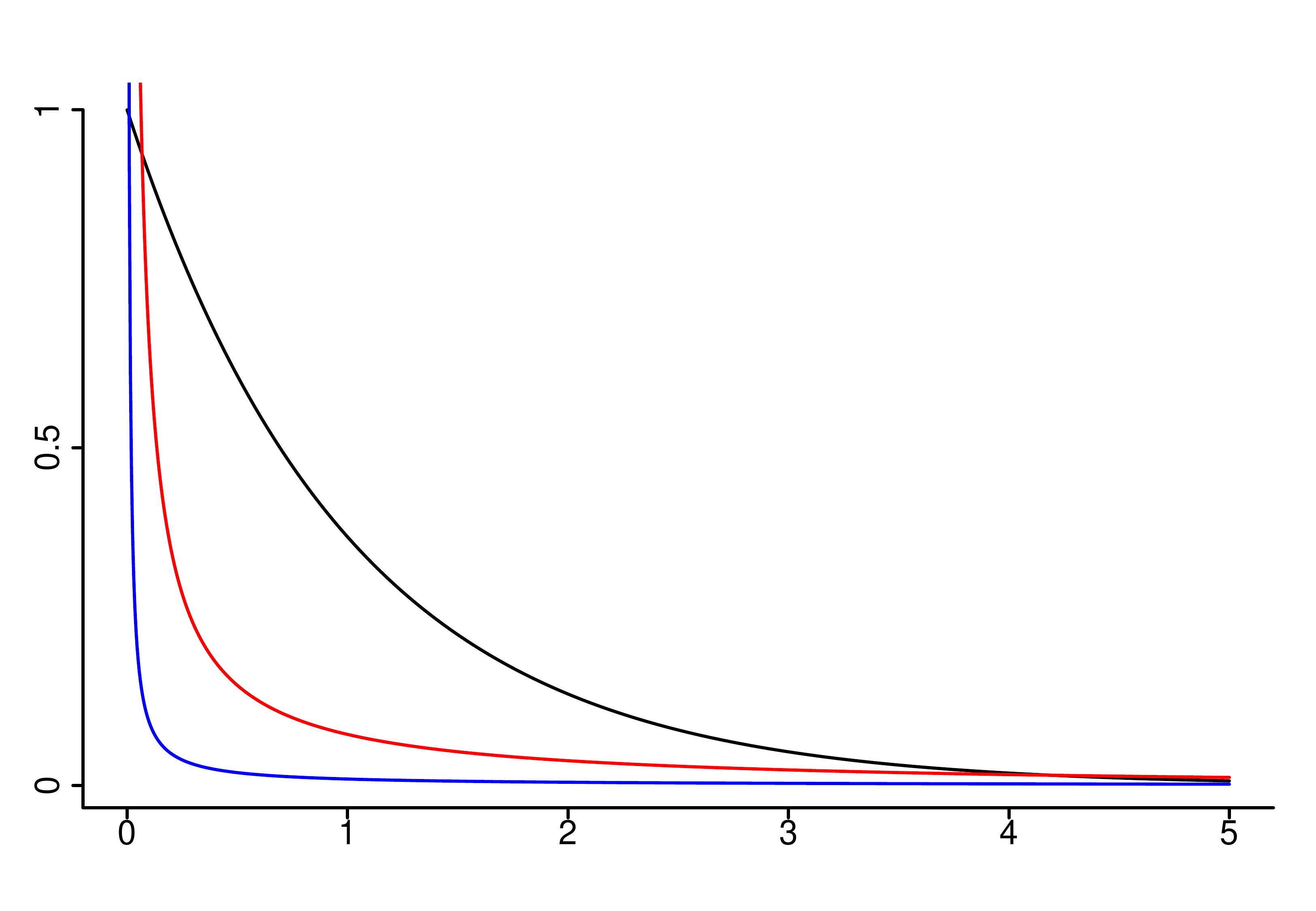}}  
  \subfloat{\label{fig:priors2}\includegraphics[width=0.35\textwidth]{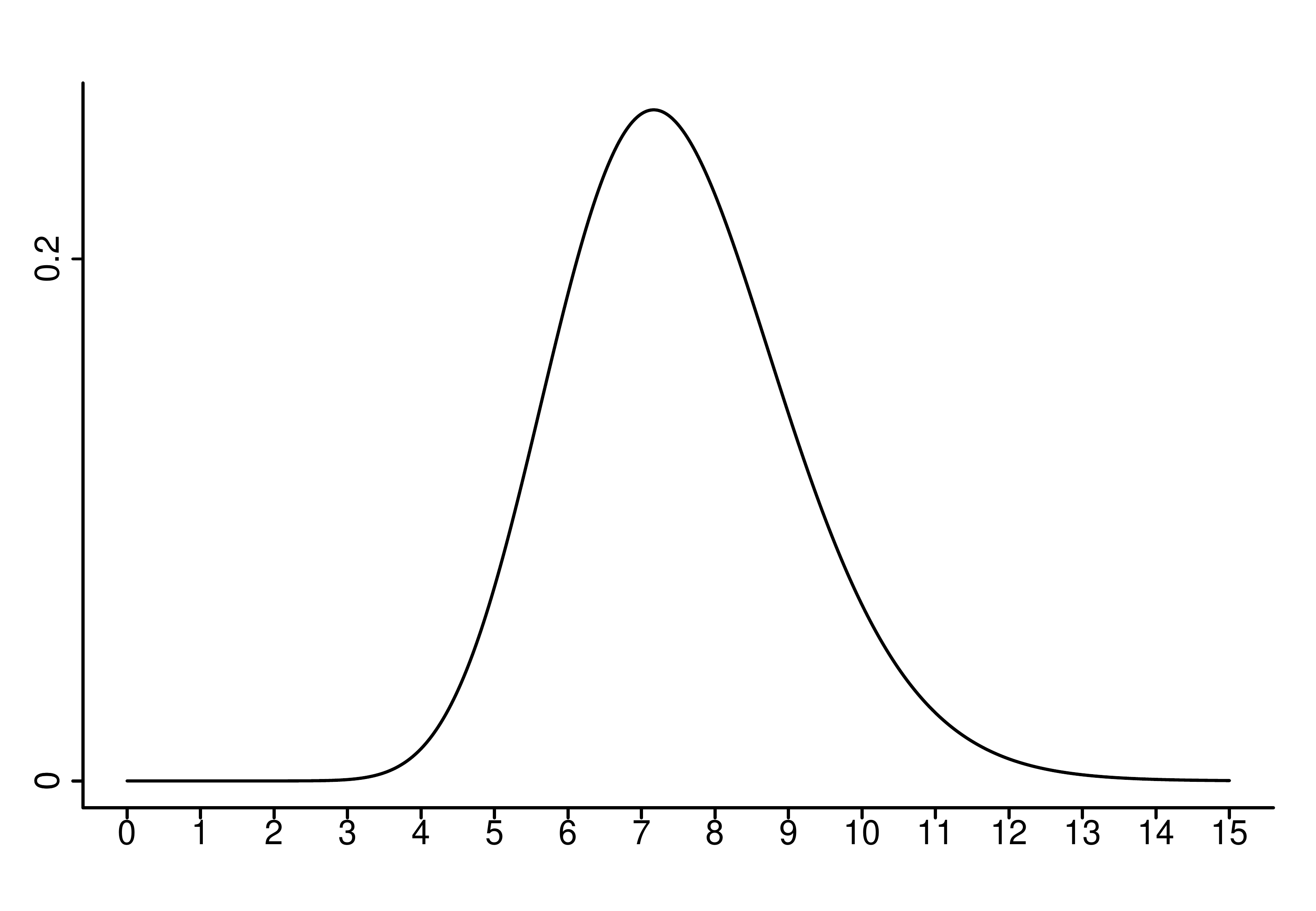}}         
\caption{The left plot shows various weakly-informative priors for the Exponential parameter $\lambda$, namely Gamma(1,1) (black line), Gamma(0.1,0.1) (red line) and Gamma(0.01,0.01) (blue line). The right plot shows the informative Gamma(22.5,7) prior which is peaked at $7.5$.}
  \label{fig:priors3}
\end{figure*}

There is also a substantial literature which approaches change detection from a Bayesian standpoint \citep{Fearnhead2007, Green1995, Chib1998}. This allows prior information about both the location of the change point, and the values of the monitored parameters within each segment, to be incorporated into the model. Although most existing Bayesian literature focuses on the non-sequential Phase I setting, recently there has been important work extending such methods to sequential detection through the use of particle filters \citep{Fearnhead2007}. We now give some general remarks on the situations for which the frequentist and Bayesian approaches are appropriate.

The method described in  \cite{Fearnhead2007} starts by putting a prior $g(x)$ on the number of observations between each pair of successive change points, with the Geometric and Negative Binomial being standard choices. Next, within each segment  a number of models $M_1,\ldots,M_k$ are allowed, with all parameters chosen in a manner such that the marginal likelihood $L(r,s,i)$ for model $i$ in the segment $x_{r+1},\ldots,x_s$ can be obtained analytically under the assumption that there are no change points within the segment - typically this is satisfied as long as conjugate priors are used. Next, a sequence of latent state variables $C_1,\ldots,C_n$ are introduced, where $C_t$ is associated with observation $x_t$ and denotes the location of the most recent change point (i.e. $C_t = k$ if the location of the last change point before $x_t$ occurred at observation $x_k$, with $C_t=0$ if no change points have occurred so far).  Under this parameterisation, the authors present a set of recursive equations which allow the posterior distributions of each $C_t$ to be computed sequentially. This allows for exact sampling from the posterior distribution of the change points. Although the computational time required to compute each posterior distribution $C_t$ increases linearly with the number of observations, making it unsuitable for sequential data sets with more than a few hundred observations, the authors present an approximation based on particle filters which achieves constant computational complexity subject to a specified approximation error. This methodology can be easily adapted to the change point problems we have considered in previous sections,.

We will use the following two simple change detection problem to illustrate the key differences between the two approaches, both based on the Exponential($\lambda$) distribution. In the first example, $\{X\}_t$ denotes a sequence of random variables where $X_t \sim \mathrm{Exp}(1)$ if $t < 50$ and $X_t \sim \mathrm{Exp}(3)$ otherwise. In the second, $\{Y\}_t$ denotes a sequence where $Y_t \sim \mathrm{Exp}(5)$ if $t < 50$ and $Y_t \sim \mathrm{Exp}(10)$ otherwise. Note that we have chosen to have only a small number of observations prior to the change point so that the posterior distributions from \cite{Fearnhead2007}  can be computed exactly without need for a particle approximation, and we restrict attention to sequences containing only a single change point in order to keep the analysis simple.

Suppose first that very little prior information is available regarding either the change point location, or the values of the Exponential distribution parameter $\lambda$ within each segment, and so relatively non-informative priors must be chosen. For the segment length, a Negative Binomial prior with mean $200$ and standard deviation $200$ is chosen. Selecting a non informative prior for the  $\lambda$ parameters is more difficult. Because change detection is essentially a model selection problem, using a prior which is overly non-informative makes it more difficult to detect change points, a consequence of Lindley's paradox \citep{Bernardo2000}. We will use a conjugate Gamma$(\epsilon,\epsilon)$ prior where the prior becomes flatter as $\epsilon \rightarrow 0$, with Gamma$(1/2,0)$ being the (improper) Jeffrey's prior. Figure \ref{fig:priors} shows a plot of these priors for various choices of $\epsilon$.

In order to use the Bayesian scheme for sequential change detection, we flag that a change has occurred at time $\hat{\tau}$ where $\hat{\tau} = \min_t P(C_t  =  0| x_1,\ldots,x_t) < c$, i.e. the first time that the probability of there being no previous change point drops below some fixed value $c$. In practice $c$ will be chosen in order to minimise a loss function which trades-off the cost of false alarms against quick detections but, to allow direct comparison to the frequentist approach, we consider a range of values of $c$. For each choice of $c$, we simulated 10000 realisations of both $\{X\}_t$ and $\{Y\}_t$, and computed both the proportion of times a false positive was generated (defined as a change being signalled before observation $50$), and the average number of observations taken for a change to be detected, conditional on $\hat{\tau} > 50$. We also performed the same set of experiments using the frequentist test statistic $M_t^c$, where we chose the $ARL_0$ to be $200$ in order to match the expected value of the Bayesian prior on the segment length.

Table \ref{tab:exp1} shows the average results for the $\{X\}_t$ sequences where the distribution changes from Exp(1) to Exp(3) after $50$ observations. There are several features of this which deserve comment. First, when the non-informative (and improper) Jeffrey's prior is used, the Bayesian scheme fails to detect any change points. This is a well known issue related to Lindley's paradox in model selection and illustrates the problems which come with using non-informative priors in change point problems; as the prior becomes increasingly flat, the range of values which receive non-neglible prior weight increases, resulting in an increasingly diffuse posterior which makes it hard to find change points. Similarly when $\epsilon \rightarrow 0$ in the Gamma$(\epsilon,\epsilon)$ prior corresponding to a flatter prior, the Bayesian scheme requires more and more post-change observations before the change point is detected, and the frequentist method consistently detects changes faster.

Note however that for the Gamma(1,1) prior there are  values of $c$ that result in the Bayesian scheme having both a superior false positive rate and detection speed compared to the frequentist method. This is because such a prior is quite informative and assigns a relatively high probability mass to values of $\lambda$ around $1$, as can be seen from Figure \ref{fig:priors}. To illustrate this, we repeated the analysis for the $\{Y\}_t$ sequences which change from  Exp$(5)$ to Exp$(10)$, and the results are shown in Table \ref{tab:exp2}, Here it can be seen that such a Gamma$(1,1)$ prior results in extremely slow detection of changes  due to the low weight it puts on parameter values around $\lambda=5$. Again, using a relatively non-informative Gamma (0.01,0.01) prior results in slow detections compared to the frequentist approach.

These examples show that it is quite difficult to design a Bayesian change detection scheme when there is no accurate prior information due to the difficulties encountered with non-informative priors. In a typical Bayesian modeling scenario, the solution to the above issues would be to instead use a hierarchal prior specification with a Gamma$(\alpha,\beta)$ prior being assigned to $\lambda$ in each segment, with a further prior assigned to $\alpha$ and $\beta$ to allow them to be learned from the data. However although such an approach is possible when working with non-sequential Phase I change detection problems, it is not possible in the sequential context presented in \cite{Fearnhead2007}, which requires the independence of observations in different segments and hence fixed prior paramours.

Of course, in many situations there will be accurate prior information available, and an informative rather than non-informative prior can be used. To illustrate this, we also analyzed the $\{Y\}_t$ sequence using an informative Gamma$(22.5,3)$ prior which has a mean of $7.5$, midway between the pre- and post-change values of $\lambda$, and a variance of 2.5, as shown in Figure \ref{fig:priors2}. The results when using this prior are also given in Table \ref{tab:exp2} and it can be seen that for some choices of the threshold $c$, the performance  is   superior to the frequentist method, both in terms of false positives and detection speed. This illustrates the strength of Bayesian change detection when it is possible to choose an informative prior which assigns relatively high weight to the true parameter values.

To further highlight the difference between the frequentist and Bayesian procedure, Table \ref{tab:finalexperiments} shows the false positive rate and detection delay when the pre- and post-change values of $\lambda$ are not fixed, but are sampled from the Gamma($22.5,3)$ prior. In this case, the change point occurs at $\tau=50$, and the observations are distributed as Exp($\lambda_0$) for $t \leq 50$ and as Exp($\lambda_0$) for $t > 50$, where both $\lambda_0$ and $\lambda_1$ are sampled from the Gamma($22.5,3$) distribution. Again, 10000 simulations were carried out, with different values of $\lambda_0$ and $\lambda_1$ sampled for each simulation. The prior for the Bayesian method was set to be equal to the true Gamma distribution used to sample these parameters. This is the optimal setting for the Bayesian approach, since in this case the informative prior matches the data generating distribution exactly. As can be seen from Table  \ref{tab:finalexperiments} , the Bayesian approach is unsurprisingly superior to the frequentist method in this context. 

In summary, Bayesian change detection schemes such as the one described in \cite{Fearnhead2007} give excellent performance in situations where there is enough prior knowledge about the distribution parameters in each segment to allow a relatively informative prior distribution to be specified. In these cases, performance will generally be superior to frequentist methods. However in situations where there is no such information, using a non-informative prior can result in very poor performance. In this case, the frequentist method may be preferred as a more robust alternative, with the added benefit of being able to put a hard bound on the false positive rate, even when the distributional parameters are unknown.

\begin{table}[t]
\caption{Proportion of false positives, and average detection delays for different choices of the Bayesian prior, and the frequentist $M^c_t$ statistic (top line) with an $ARL_0$ of 200}
\begin{center}
\subfloat[Exp(1) $\rightarrow$ Exp(3)]{
\begin{tabular}{rrr}
  \hline
   & Fps & Delay \\
  \hline
  $M^c_t$&0.14 & 60.5\\
  \hline
c &   \multicolumn{2}{c}{Gamma$(1,1)$}\\
   \hline
  0.2 & 0.01 & 68.7 \\
  0.4 & 0.04 & 64.9 \\
  0.6 & 0.11 & 61.7 \\
  0.8 & 0.34 & 57.7 \\
  \hline
 c&  \multicolumn{2}{c}{Gamma$(0.1,0.1)$}\\
   \hline
  0.2 & 0.01 & 72.2 \\
  0.4 & 0.02 & 68.3 \\
  0.6 & 0.04 & 65.1\\
  0.8 & 0.12 & 61.4\\
   \hline
c&   \multicolumn{2}{c}{Gamma$(0.01,0.01)$}\\
   \hline
  0.2 & 0.01& 80.9 \\
  0.4 & 0.01& 76.4\\
  0.6 & 0.02& 72.9\\
  0.8 & 0.02& 68.9 \\
\label{tab:exp1}
\end{tabular}
}
\subfloat[Exp(5) $\rightarrow$ Exp(10)]{
\begin{tabular}{rrr}
  \hline
   & Fps & Delay \\
  \hline
  $M^c_t$&0.14 & 76.1\\
  \hline
 c&   \multicolumn{2}{c}{Gamma$(1,1)$}\\
   \hline
  0.2 & 0.01 & 186.7\\
  0.4 & 0.01 & 169.3\\
  0.6 & 0.01& 153.5\\
  0.8 & 0.04 & 133.2\\
  \hline
c&   \multicolumn{2}{c}{Gamma$(0.01,0.01)$}\\
   \hline
  0.2 &0.01 & 170.0\\
  0.4 & 0.01& 151.2\\
  0.6 & 0.01& 134.4\\
  0.8 & 0.01& 114.6\\
  \hline
   c&  \multicolumn{2}{c}{Gamma$(22.5,3)$}\\
   \hline
  0.2 & 0.01& 99.3. \\
  0.4 & 0.01& 81.0\\
  0.6 & 0.02& 68.0\\
  0.8 & 0.43& 55.5\\
\label{tab:exp2}
\end{tabular}
}
\end{center}
\label{tab:exp}
\end{table}

\subsection{Real Data}
\label{sec:experiments3}

We conclude with two examples of change detection in real applications. We first consider a hard bake process which is taken from the statistical process control literature, and then look at an example using a potentially heavy-tailed financial return series.

\subsection{Hard Bake Process}
The hard-bake process is a commonly used step in the manufacturing of semi-conductors. It is typical to apply it to wafers after they have had light-sensitive photoresistive material applied, in order to increase their resist adherence and etch resistance. During this process, a key quality characteristic is the flow width of the resist. Data taken from such a process is given in \cite{Montgomery2005}, which consists of 125 flow width measurements representing the initialization phase where control chart parameters are learned, followed by 95 Phase II observations which must be monitored for changes. The key advantage of the unknown parameter formulation we have used in this paper is that there is no need to treat the initialization phase different from the actual monitoring, and so we treat the data as being a single sequence containing 220 observations. This observations are plotted in Figure \ref{fig:bake}

\begin{figure}[!t]
  \centering     
  \includegraphics[width=0.4\textwidth]{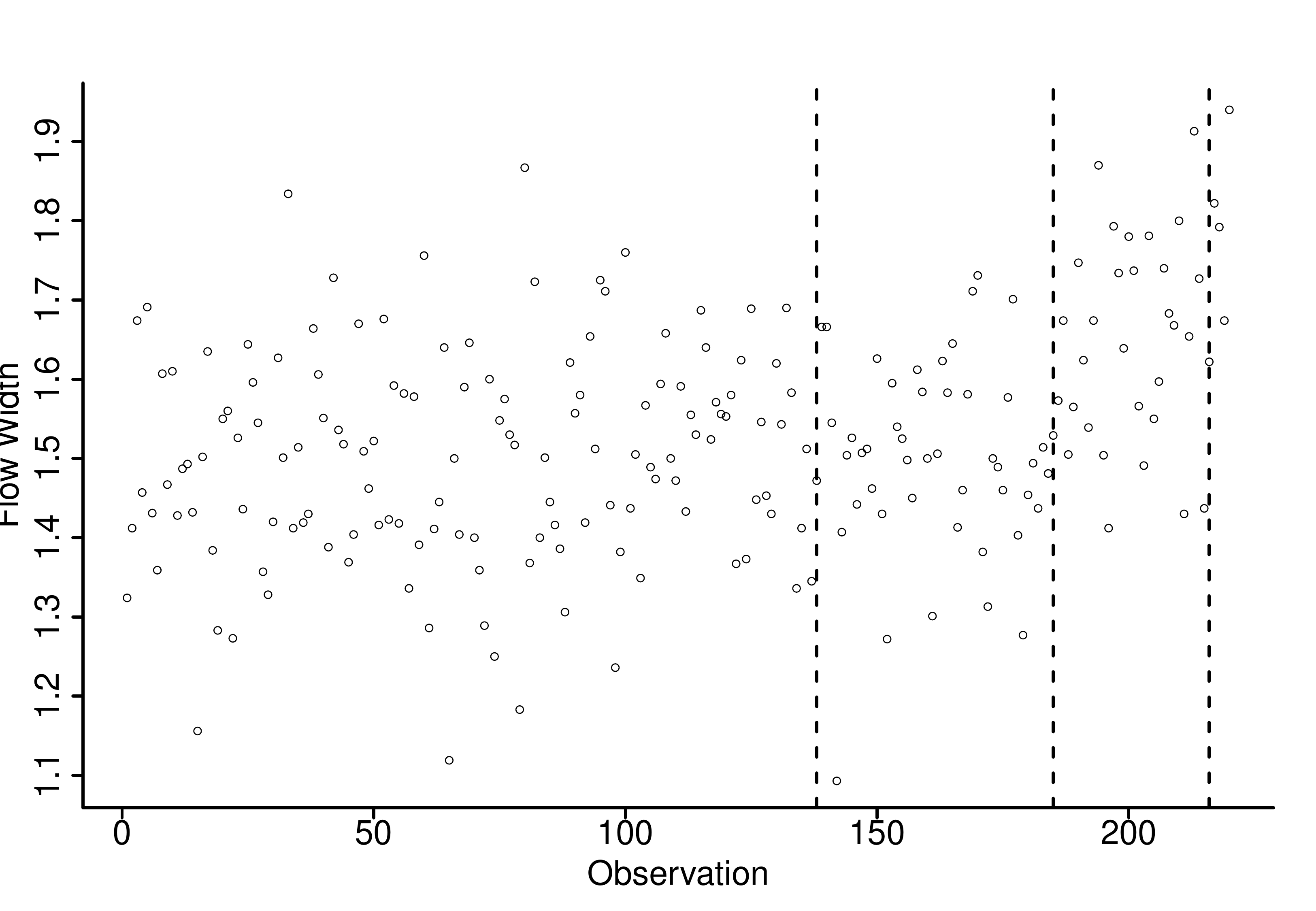}
\caption{Measurements taken from the hard bake process,  with the discovered change points superimposed as dotted lines}
  \label{fig:bake}
\end{figure}

We use the $D_t^c$ statistic in order to perform sequential change detection on this sequence. Although we previously only considered sequences containing a single change point,, the extension to multiple change points is simple. The  change detector processes the observations sequentially, until  the first time $D_t^c > h_t$, in which case a change is signaled. Suppose this occurs at observation $T_1$. Then, the best estimate of the location of the change point is $\hat{\tau}_1 \leq T_1$ where $\hat{\tau_1}$ is equal to the value of $k$ which maximized $D_{k,T_1}$. Sequential change detection then resumes at observation $X_{\hat{\tau}_1 + 1}$, which is the first observation following the estimated change point, with the previous observations being discarded.

Using an $ARL_0$ of $500$ as in previous examples,  change point signals were given at observations $140, 204$, and 221. The estimated change point locations were at observations $138, 185$, and $216$ respectively, and these are plotted on Figure \ref{fig:bake}. Repeating the same analysis using the $H_t$ statistic from \cite{Hawkins2005a} with the asymptotic correction resulted in the same three change point locations being estimated, but with corresponding detection times $140, 207$, and 221. If we can assume that the estimated change point locations correspond to true changes, then this implies that the second change point was detected $3$ observations faster when using the corrected $D_t^c$ statistic, which is consistent with our previous findings from Section \ref{sec:experiments1}. If these change points represent genuine faults in the underlying manufacturing process, then the ability to detect them faster may be important in practice, if it allows corrective action to be taken earlier.

\subsection{Financial Data}
Since the change point models we have considered are parametric, appropriate care must be taken when deploying them to ensure that the parametric assumptions are satisfied. Failure to do so may result in spurious false positives, and/or slow change detections. We illustrate this with an example using financial stock market returns, where we analyze the weekly returns of the Dow Jones Industrial Average stock market index. This index is widely traded, and is made up of $30$ large publicly owned American companies The data we have spans the period from the $1^{st}$ of January 1991 to the $31^{st}$ of October 2011. Let $P_t$ denote the opening price of the Dow Jones stock market index on week $t$ for each of the $1053$ weeks in the sample period The log returns $X_t = \log(P_t/P_{t-1})$ then represent the weekly price changes, and this series is plotted in Figure \ref{fig:dow1} where it can clearly be seen that the variance of this series is non-stationary and changes over time. Finding appropriate models for such return series is a widely studied problem within financial econometrics. Although the conditional variance of financial returns is often modelled using a pure GARCH process \citep{Engle2001} it is now accepted that many return series also contain structural breaks in the unconditional variance, and that these can be found using a change point approach . It is common to first look for change points by treating the $X_t$ series as if it the observations are independent and identically distributed between each pair of change points \citep{Aggarwal1999,Rossphysica,Inclan1994}.

\begin{figure*}[!t]
  \centering     
  \subfloat[Parametric]{\label{fig:dow1}\includegraphics[width=0.5\textwidth]{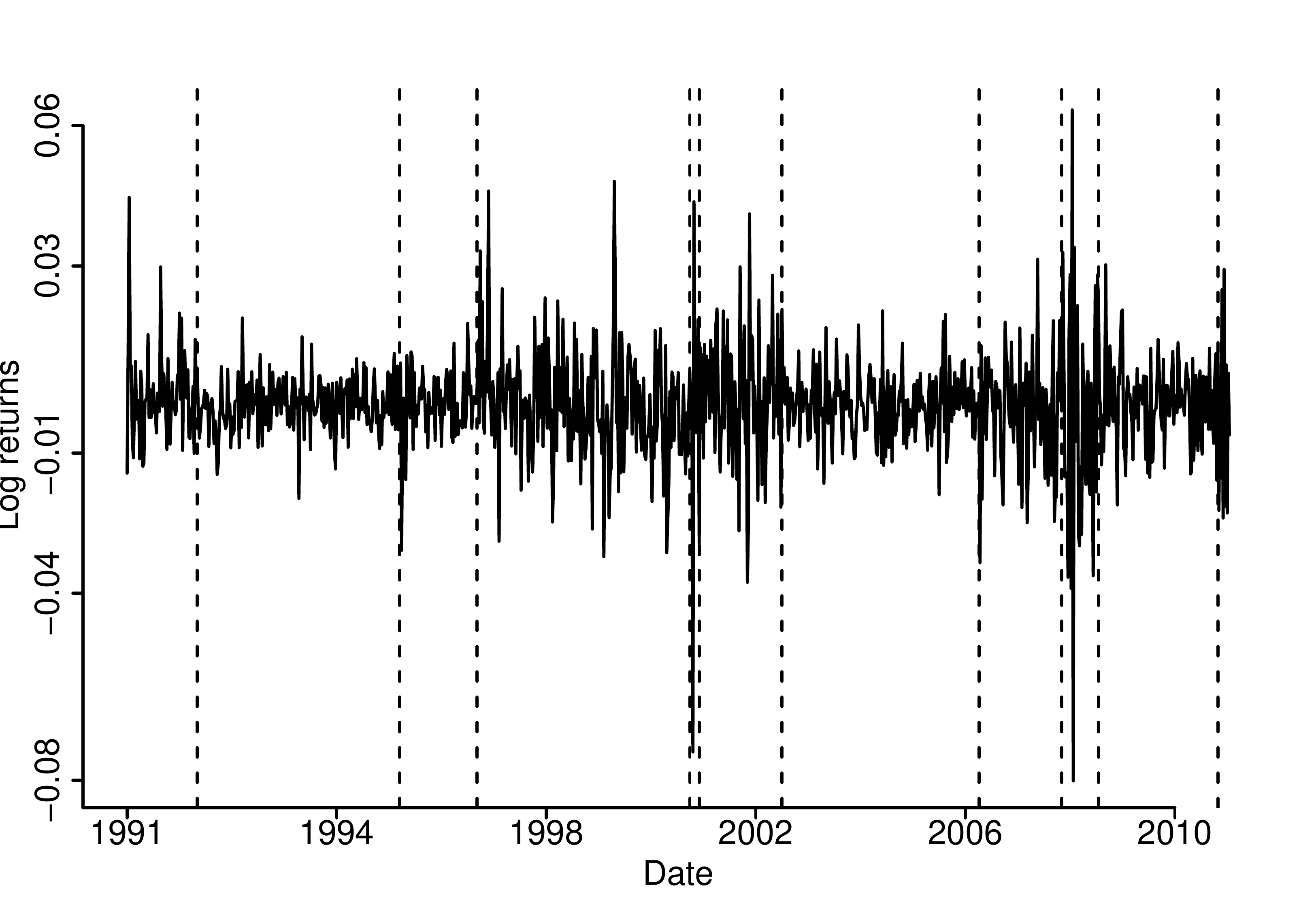}}
  \subfloat[Nonparametric]{\label{fig:dow2}\includegraphics[width=0.5\textwidth]{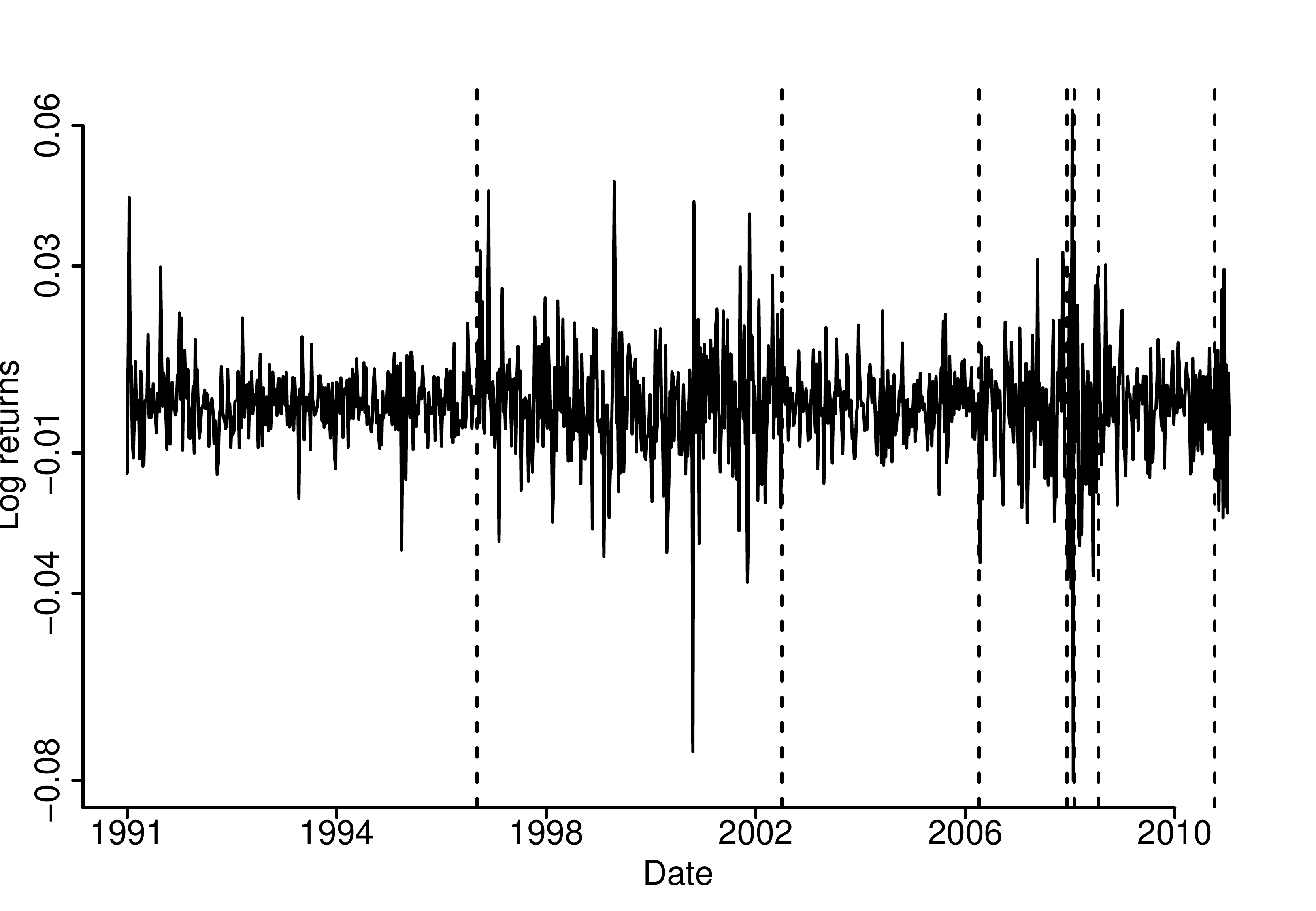}}           
\caption{Log-returns of the Dow Jones stock index from $2001$ to $2011$, with the change points found using both the parametric $D_t^c$ statistic and the nonparametric Lepage statistic superimposed as dotted lines}
  \label{fig:priceseries}
\end{figure*}

We illustrate this by using the $D_t^c$ statistic to sequentially locate the change points in this sequence of returns. This is done an identical manner to the previous example, where the sequence is processed one observation at a time and a sequential decision is made about whether a change has occurred after each data point. Since there are a large number of observations, we chose an $ARL_0$ of 5000 to avoid an excessive number of false positives being generated. In total $10$ change points were detected, which are shown in Figure \ref{fig:dow1}.  Running the same analysis using the asymptotically corrected $H_t$ statistic from \cite{Hawkins2005a} again resulted in the same $10$ change points being found, but as in the previous examples, $4$ of these were detected later than when using the finite sample corrected statistic $D_t^c$ with the increased delays varying from between 1 and 5 additional observations.

The likelihood ratio test underlying both the $D_t^c$ and $H_t$ statistics is based on the assumption of Gaussianity. However previous studies have found  evidence that financial returns are non-Gaussian and exhibit heavy tail behavior, even when the conditional variance is taken to be time-varying \citep{Rossphysica}. To investigate this, we also tried detecting change points using the nonparametric Lepage-based statistic described in \cite{Rosstechnometrics} which can detect change points in the mean and/or variance without making distributional assumptions. Running this method with the $ARL_0$ also set to $5000$ resulted in only 7 change points being detected, which are shown in Figure \ref{fig:dow2}. Comparing the change points found by the two methods, it can be seen that the Gaussian assumption made when using the $D_t^c$ statistic results in extreme outlying observations being interpreted as change points, such as the change point found in late $1995$. As the nonparametric test is more robust against heavy tailed data, it detects fewer change points.  This highlights the pitfalls which can arise when deploying such parametric models to series which are not known to be Gaussian. In this case, it may be more appropriate to first make a transformation to the data in order to make it closer to Gaussian \citep{Qiu2011}, or simply use a nonparametric technique.

Of course, in situations where the parametric assumptions used in the likelihood ratio test are correct, the parametric change point models will generally be able to detect changes faster than their nonparametric counterparts, conditional on the same bound on the $ARL_0$. 

\section{Concluding Remarks}
The task of sequential change detection is more difficult when the parameters of the pre- and post-change distributions are unknown and must be estimated from the data. In this situation, the sequential generalised likelihood ratio testing gives a computationally efficient procedure which is able to achieve a desired bound on the rate at which false positives occur. In \cite{Hawkins2005a}, an approach is developed along these lines, however we have shown that it suffers from a small but persistent bias due to the way in which test statistics are calculated. Their procedure relies on an asymptotic argument which fails in the sequential context where samples containing only a small number of observations must be used, regardless of how many observations are available. By introducing a finite sample correction for this statistic, we have given a more powerful version of their method which is able to detect changes in both mean and variance faster, across all combinations of change magnitude and location. The extra computational burden introduced by our approach is minimal, and consists only of modifying the boundary statistics by subtraction and division by a constant, and should therefore be preferred when performing Gaussian change detection in practice. We also showed that such issues can arise when performing sequential parametric change detection using other distributional forms, which we illustrated by constructing a novel change detection procedure for the Exponential distribution. As in the Gaussian case, the use of a finite-sample correction provides identical or better performance in all situations.

\section{Supplementary Material}
R code implementing the change detection algorithm using the $D^c_t$ and $M^c_t$ statistics is contained in the \textbf{cpm} R package available from CRAN: \url{http://cran.r-project.org/web/packages/cpm/index.html}. 

Documentation for this package can be obtained from the author's website: \url{http://www.gordonjross.co.uk/software.html}

\appendix

\section{Test Statistic Moments}
\newtheorem{thm}{Theorem}
 \begin{thm} 
Under the null hypothesis of no change:
\begin{align*}
E[D&_{k,t}] =  t (\log(2/t) + \psi( (t-1)/2)) \\& - k(\log(2/k) + \psi((k-1)/2)) \\& - (t-k)(\log(2/(t-k)) + \psi((t-k-1)/2)).
\end{align*}
where $\psi$ denotes the digamma function.

\end{thm} 

\begin{proof}
We roughly follow the argument of  \cite{Zamba2009a}.  From Equation \ref{eqn:lrt}, we have that 
\begin{align}
E[D&_{k,t}] =   k E[\log S_{0,t}] - k E[\log S_{0,kj}] +  (t-k) E[S_{0,t}]\notag \\& - (t-k) E[S_{k,n}]  .
\label{eqn:proof}
\end{align}
Consider the first term $E[\log S_{0,t}] $. By basic properties of the Gaussian distribution $S_{0,t}$ has the same distribution as a $\chi^2_{t-1}$ random variable multiplied by a factor of $K = t/(t-1)\sigma^2$ where $\sigma^2$ is the true variance. Let $W_{t-1} \sim \chi^2_{t-1}$, then 
\[E[\log(S_{0,t})] \sim \log(K) + E[\log(W_{t-1})].\]
Now, $E[\log(W_{t-1})]$ can be calculated using the moment generating function, $M_{\log W_{t-1}}$. Note that $M_{\log W_{t-1}}(t) = E[W_{t-1}]$. Computing this expectation and then differentiating the moment generating function yields $E[W_{t-1}] =  \log 2 + \psi((t-1)/2)$. Repeating this argument for the other terms in Equation \ref{eqn:proof} gives the desired result, with the $K$ factors cancelling out.
\end{proof}



 \begin{thm} 
Under the null hypothesis of no change:
\begin{align*}
E[M&_{k,t}] = -2[k\psi(k)  + (n-k)\psi(n-k) -n\psi(n) + \\& n \log (n) - k \log k - (n-k) \log (n-k)]
\end{align*}
where $\psi$ denotes the digamma function.
\end{thm}

\begin{proof}
First note that if $Y_1,\ldots,Y_m$ are i.i.d $\mathrm{Exp}(\lambda)$ random variables then $\sum_{i=1}^{m}Y_i \sim \mathrm{Gamma}(n,\lambda)$.
By separating out the log terms, it follows that $E[M_{k,t}]$ has the same distribution as:

\begin{align*}
-2[&-n \log (V+W) + k \log V + (n-k) \log W  + n \log (n) \\&- k \log k - (n-k) \log (n-k)]
\end{align*}
where $V \sim \mathrm{Gamma}(k,\lambda)$ and $W \sim \mathrm{Gamma}(n-k,\lambda)$. Using a similar argument to the above proof based on the moment generation function, it can easily be shown that $E[\log V] = \psi(k) - \log(\lambda)$. The result then follows, with the $\lambda$ terms canceling out.
\end{proof}

\section{Results Using an Informative Prior}
Table \ref{tab:finalexperiments} shows the false positives and average detection delay in the ideal case described in Section \ref{sec:experiments2} where the parameters in the Bayesian change detection model are assigned at Gamma($22.5,3$) prior, and the parameters used in the simulation are also sampled from this distribution.

\begin{table}[h]
\caption{Proportion of false positives, and average detection delays when parameters are simulated from the Gamma(22.5,3) prior}
\begin{center}
\begin{tabular}{rrr}
  \hline
   & Fps & Delay \\
  \hline
  $M^c_t$&0.34 & 153.03\\
  \hline
c &   \multicolumn{2}{c}{Gamma$(22.5,3)$}\\
   \hline
  0.2 & $<$0.01 & 207.67 \\
  0.4 & $<$0.01 & 114.89 \\
  0.6 & 0.11 & 49.47 \\
  0.8 & 0.63 & 6.20 \\
  \label{tab:exp1}
\end{tabular}
\end{center}
\label{tab:finalexperiments}
\end{table}

\section{ $h_t$ Thresholds}
Table \ref{tab:thresholds} gives values of the exponentially smoothed threshold sequences $\tilde{h}_t$ which correspond to several choices of the $ARL_0$. Note that these thresholds become roughly constant (subject to sampling variation) after a few hundred observations, and so for values of $t>800$, using the threshold value corresponding to $t=800$ is advised. For the case where the $ARL_0=100$, computing the thresholds for high values of $t$ is very computationally expensive, and so the thresholds were only computed up to $t=500$, with values above this being computed using spline interpolation.

\begin{table*}[t]
\caption{Values of the threshold sequences $h_t$ corresponding to various choices of the $ARL_0$ when monitoring Gaussian ($D^c_t$) and Exponential ($M^c_t$) sequences.} 
\begin{center}
\subfloat[$D^c_t$]{
\begin{tabular}{rrrrrrrr}
  \hline
t & 100 & 200 & 370 & 500 & 1000 & 2000 & 5000 \\ 
  \hline
21 & 13.2 & 14.8 & 16.1 & 16.8 & 18.1 & 19.7 & 21.5 \\ 
  22 & 13.1 & 14.7 & 16.0 & 16.7 & 18.0 & 19.6 & 21.5 \\ 
  23 & 13.0 & 14.6 & 15.9 & 16.6 & 18.0 & 19.6 & 21.4 \\ 
  24 & 12.9 & 14.5 & 15.8 & 16.5 & 17.9 & 19.5 & 21.4 \\ 
  25 & 12.8 & 14.3 & 15.7 & 16.4 & 17.8 & 19.4 & 21.3 \\ 
  26 & 12.7 & 14.3 & 15.7 & 16.3 & 17.8 & 19.3 & 21.2 \\ 
  27 & 12.6 & 14.2 & 15.6 & 16.2 & 17.7 & 19.2 & 21.2 \\ 
  28 & 12.5 & 14.1 & 15.5 & 16.2 & 17.6 & 19.2 & 21.1 \\ 
  29 & 12.5 & 14.1 & 15.5 & 16.2 & 17.6 & 19.2 & 21.0 \\ 
  30 & 12.4 & 14.0 & 15.5 & 16.2 & 17.6 & 19.2 & 21.0 \\ 
  50 & 12.3 & 13.9 & 15.4 & 16.1 & 17.7 & 19.3 & 21.2 \\ 
  60 & 12.4 & 14.0 & 15.5 & 16.2 & 17.8 & 19.3 & 21.3 \\ 
  80 & 12.3 & 14.1 & 15.5 & 16.2 & 17.8 & 19.4 & 21.4 \\ 
  100 & 12.4 & 14.1 & 15.5 & 16.3 & 17.9 & 19.4 & 21.6 \\ 
  200 & 12.4 & 14.1 & 15.6 & 16.4 & 18.0 & 19.6 & 21.6 \\ 
  300 & 12.4 & 14.1 & 15.7 & 16.4 & 18.0 & 19.6 & 21.5 \\ 
  400 & 12.1 & 14.0 & 15.6 & 16.3 & 18.0 & 19.7 & 21.8 \\ 
  500 & 12.2 & 14.2 & 15.7 & 16.4 & 18.0 & 19.6 & 21.7 \\ 
  600 & 12.3 & 14.1 & 15.6 & 16.4 & 18.1 & 19.7 & 21.8 \\ 
  700 & 12.3 & 14.3 & 15.6 & 16.4 & 18.0 & 19.6 & 21.7 \\ 
  800 & 12.3 & 14.1 & 15.6 & 16.3 & 18.0 & 19.6 & 21.7 \\ 
   \hline
\end{tabular}
}
\subfloat[$M^c_t$]{
\begin{tabular}{rrrrrrrr}
  \hline
t & 100 & 200 & 370 & 500 & 1000 & 2000 & 5000 \\ 
  \hline
21 & 5.2 & 5.9 & 6.5 & 6.8 & 7.4 & 8.0 & 8.9 \\ 
  22 & 5.1 & 5.8 & 6.4 & 6.7 & 7.3 & 7.9 & 8.8 \\ 
  23 & 5.0 & 5.6 & 6.2 & 6.5 & 7.2 & 7.8 & 8.7 \\ 
  24 & 4.8 & 5.5 & 6.1 & 6.4 & 7.1 & 7.7 & 8.6 \\ 
  25 & 4.7 & 5.4 & 6.0 & 6.3 & 7.0 & 7.7 & 8.5 \\ 
  26 & 4.6 & 5.3 & 5.9 & 6.2 & 6.9 & 7.6 & 8.4 \\ 
  27 & 4.5 & 5.2 & 5.8 & 6.1 & 6.8 & 7.5 & 8.4 \\ 
  28 & 4.4 & 5.1 & 5.8 & 6.1 & 6.7 & 7.4 & 8.3 \\ 
  29 & 4.4 & 5.1 & 5.7 & 6.0 & 6.7 & 7.4 & 8.3 \\ 
  30 & 4.3 & 5.0 & 5.7 & 6.0 & 6.7 & 7.4 & 8.3 \\ 
  50 & 4.0 & 4.8 & 5.5 & 5.8 & 6.5 & 7.2 & 8.2 \\ 
  60 & 4.0 & 4.8 & 5.5 & 5.8 & 6.5 & 7.3 & 8.2 \\ 
  80 & 4.0 & 4.8 & 5.5 & 5.8 & 6.6 & 7.3 & 8.2 \\ 
  100 & 4.1 & 4.9 & 5.6 & 5.9 & 6.6 & 7.4 & 8.3 \\ 
  200 & 4.1 & 4.9 & 5.6 & 5.9 & 6.7 & 7.4 & 8.4 \\ 
  300 & 4.0 & 4.9 & 5.6 & 5.9 & 6.6 & 7.4 & 8.4 \\ 
  400 & 4.1 & 4.8 & 5.5 & 5.9 & 6.7 & 7.5 & 8.4 \\ 
  500 & 4.1 & 4.9 & 5.5 & 5.9 & 6.7 & 7.4 & 8.4 \\ 
  600 & 4.1 & 4.8 & 5.6 & 5.9 & 6.7 & 7.5 & 8.4 \\ 
  700 & 4.1 & 4.9 & 5.5 & 5.9 & 6.7 & 7.4 & 8.4 \\ 
  800 & 4.1 & 4.8 & 5.6 & 5.9 & 6.7 & 7.4 & 8.4 \\ 
\hline
\end{tabular}
}
\end{center}
\label{tab:thresholds}
\end{table*}

\bibliographystyle{spbasic} 
\bibliography{jabref}	
\end{document}